\documentclass{article}
\usepackage{ijcai20}
\usepackage{times}

\usepackage{soul}
\usepackage{url}
\usepackage[hidelinks]{hyperref}
\usepackage[utf8]{inputenc}
\usepackage[small]{caption}
\usepackage{graphicx}
\usepackage{amsmath}
\usepackage{amsthm}
\usepackage{amssymb}
\usepackage{booktabs}
\usepackage{algorithm}
\usepackage{algorithmic}
\usepackage[colorinlistoftodos,prependcaption, disable]{todonotes}

\newtheorem{theorem}{Theorem}
\newtheorem{definition}{Definition}

\newtheorem{proposition}{Proposition}
\newtheorem{corollary}{Corollary}

\newcommand\cost[3]{c_{#1}(#2,#3)}
\newcommand\costs[2]{c_{#1}(#2)}
\newcommand\type{t}

\newcommand{\citet}[1]{\citeauthor{#1} \shortcite{#1}}

\title{The Competitive Effects of Variance-based Pricing}
\author{
Ludwig Dierks$^1$\and
Sven Seuken$^1$\\
\affiliations
$^1$Department of Informatics, University Zurich \\
\emails
       dierks@ifi.uzh.ch,
seuken@ifi.uzh.ch
}

\begin{document}

\maketitle

\begin{abstract}
In many markets, like electricity or cloud computing markets, providers incur large costs for keeping sufficient capacity in reserve to accommodate demand fluctuations of a mostly fixed user base. These costs are significantly affected by the unpredictability of the users' demand. Nevertheless, standard mechanisms charge fixed per-unit prices that do not depend on the variability of the users' demand. In this paper, we study a variance-based pricing rule in a two-provider market setting and perform a game-theoretic analysis of the resulting competitive effects. We show that an innovative provider who employs variance-based pricing can choose a pricing strategy that guarantees himself a higher profit than using fixed per-unit prices for any individually rational response of a provider playing a fixed pricing strategy. We characterize all equilibria for the setting where both providers employ variance-based pricing strategies. We find that, while in equilibrium, the profits of the providers may increase or decrease depending on their cost functions, social welfare always weakly increases.
\end{abstract}

\section{Introduction}
In most markets with mostly fixed user bases, providers' costs are largely driven by how much buffer capacity they must keep in reserve. This, in turn, depends on the variance of their users' demand. However, the predominant pricing mechanisms employed in practice  do not take this effect into account.
Instead, prices are typically set on a per-unit basis, such that every user pays the same for the same product.

In electricity markets, for example, a provider has to make long-term supply decisions. But in real-time, supply and  demand always have to be perfectly balanced, which requires a costly buffer infrastructure \cite{cramton2017electricity}. If users would always consume (almost) the same amount of energy, this buffer could be far smaller than with widely varying user demands.
Nevertheless, users pay simple per MW/h prices.

Next, consider the market for mobile data. Mobile network providers continuously expand their cell tower infrastructure to be able to satisfy their users' bandwidth needs under peak demand \cite{lopez2009ofdma}.
However, most end-users pay a fixed per-gigabyte-price, independent of when they consume it or how variable their demand is.\footnote{Even in high-GDP countries, less than 10\% of customers have unlimited data plans \cite{erics}.}

Finally, consider cloud computing markets, where cloud providers must keep buffers of idle capacity in each compute cluster to handle changing resource demands of already running jobs. Handling the variance of cloud users is particularly difficult given that  the resource needs of an individual job may vary by a factor 10 or 100 over time (see  \citet{dierks}). In this domain, the  \emph{mixture} of user types  (i.e., their average demand variance) significantly affects the provider's need to supply buffer capacity. Nevertheless, most cloud resources are sold for a fixed price per core-hour, without regard to the variability of the users' demand.  

%
%
%

\paragraph{Managing Demand via Sophisticated Pricing}
Classic approaches for dealing with varying demand include dynamic pricing and congestion-based pricing  \cite{muratori2015residential,rong2018dynamic,truong2014novel}. These approaches focus on flattening demand peaks. A big downside is that they make it unpredictable for users whether they can obtain the product at a given price when they need it. This puts providers who serve {risk-averse} users or users with relatively uniform but inelastic demand at a competitive disadvantage. In some markets, like cloud computing, this effect is so strong that providers never consider dynamic pricing for their \emph{primary} market offerings \cite{dierks2016cloud}. The effect even greatly hinders the adoption of dynamic pricing in more suitable domains, like electricity markets \cite{joskow2012dynamic}, where customers are instead often only exposed to fixed 'time-of-use' tariffs \cite{urieli2016autonomous,celebi2012time}.

\paragraph{Variance-based Pricing}
In this paper, we study \emph{variance-based pricing}, where part of the price the users pay depends on the variance of their demand.
Variance-based pricing was recently proposed by \citet{dierks} to reduce costs by improving the demand prediction ability of a monopolistic cloud provider.
For a provider, in general markets, variance-based pricing has the advantage that his low-variance users pay lower prices and are thus impacted less by the buffer requirements (which are mainly caused by high-variance users). This is not only  fairer, but importantly, it incentivizes users to reduce their variance, which in turn reduces the provider's costs. In a monopolistic setting, the provider can   obviously use variance-based pricing to increase his profits. However, in a competitive market environment, the effects are unclear, because the competitive pricing pressure by other providers may limit what he can achieve with variance-based pricing.

\paragraph{Overview of our Approach}
We analyze a duopoly of two providers  that compete for a continuum of user types. For each provider, the cost per unit of a product depends on the average variance of the users he attracts and providers either conservatively employ constant per-unit prices or are willing to innovate and employ per-unit prices that linearly depend on the user's variance. We restrict ourself to linear prices here, because their simplicity makes them most plausible and marketable in practice.
We show that, as long as a provider's costs are not far larger than the costs of his competitor, unilaterally switching to variance-based pricing can be used to obtain a higher profit for any reasonable constant\ response of the other provider. We characterize all equilibria that arise if both providers employ variance-based pricing. We also show that, as long as providers are not symmetric in their cost functions, the profits of both providers often increase, as they can  attract user types that their cost function is better suited to serve. Finally, we show that the welfare may decrease if only one provider employs variance-based pricing, but that it can only weakly increase if both employ variance-based pricing. 
 
Variance-based pricing is a type of price discrimination \cite{varian1989price,mussa}. However, in contrast to the classic price discrimination settings \cite{moorthy1984market,blattberg1989price,gallego2006price}, variance-based pricing does not discriminate based on user preferences. Similarly, in our model there is  neither a fixed marginal cost of supplying a given user nor do costs depend solely on the number of supplied products. As long as neither provider charges for variance, they cannot price discriminate at all and the problem becomes similar to a Bertrand competition  \cite{baye2017}. In order to isolate the competitive effects caused by variance-based pricing, we assume that both providers offer the same product; thus, product differentiation (e.g.,  \cite{feng2013price}) does not take place.
%
%
%
%

\section{Preliminaries}
\todo[inline]{indexing}
We consider a market setting with two providers and a continuum of users to which the providers want to sell their products.
We  study a simple two-period model: in the first period, the providers choose their pricing strategies; in the second period,  the users choose a provider to buy the product from.

\subsection{Formal Model}

Each user is associated with a real-valued type $\type\in [0, \type_{max}]$. We keep this type general, but in practice it can be assumed to encode the variance of a user's demand. 
To keep the notations simple, we normalize each user's expected demand to $1$. For a randomly chosen user, her type is distributed as a continuous random variable with pdf $f(\type)$ and cdf $F(\type)$. We do not model fluctuations in the mixture of user types over time. 

Each provider's strategy space consists of his choice of price function $\rho_i$. We restrict $\rho_i=(p_i^f,p_i^\ell)$ to a \emph{fixed} price $p_i^f$ per unit of the product (independent of the user type)  and a second \emph{linearly} type-based charge $p_i^\ell \type$ per unit of the product.\footnote{ As supply decisions, and therefore costs, do not depend on a user's true type, but on the providers predicted type, we assume that providers know each realized users type. In practice, user bases are often mostly fixed, so provider can learn a users type over time.} The overall payment per unit of the product for a user with type $\type$ is given by $\rho_{i}(\type) =p_i^f +  p_i^\ell \type$. Going forward, we refer to providers that are willing to choose any  $p_i^f, p_i^\ell \in \mathbb{R}$ as \emph{innovative} and those who do not adopt the new pricing scheme and restrict themselves to $p_i^\ell =0$ as \emph{conservative}. 

Given the provider's price functions, and depending on a user's own type, each user chooses to obtain the product either from provider $1$ or $2$. 
We denote a strategy profile for all user types by $\sigma(\type):[0,\type_{max}]\rightarrow \left\lbrace 1, 2 \right\rbrace$.
In this paper, we do not model the users' values for product consumption.\footnote{In the markets we study, essentially every user is served by some provider, as costs are very low compared to most users' values and competition ensures that prices are close to costs.} Consequently, the users' utility function is simply equal to their negative payments. Throughout the paper, given price functions $(\rho_1, \rho_2)$, we assume that users  only play \emph{utility-maximizing} (i.e., payment minimizing) user strategy profiles. Formally, we assume that $\sigma(\type)= argmin_i \left\lbrace          \rho_{i}(\type) \right\rbrace$ for all $\type$.
  
As we will see, it is often in a provider's best interest to play essentially the same strategy as their opponent, which makes tie-breaking rules for the user sub-game very important. 
To avoid that tie breaking for $\rho_1=\rho_2$ gives rise to arbitrary user strategy profiles that would not arise in practice, we restrict user strategy profiles to those that arise as the limit of uniquely (up to a null-set) utility-maximizing user strategy profiles for some sequence of pricing functions converging to $(\rho_1,\rho_2)$. We also call these user strategy profiles \emph{enforceable}, as at least one provider can enforce them by deviating from $(\rho_1,\rho_2)$ at only an infinitesimal profit loss. 
With linearly type-based prices, this limits strategy profiles to any form where all users with type greater than some cutoff point $\hat{\type}$ join one provider $i$ and those with lower type join the other provider. To denote this, we also write $\sigma = [0,\hat{\type}]_{\rightarrow i}$ or $\sigma =  [\hat{\type},\type_{max}]_{\rightarrow i}$. Note that any $\sigma$ is uniquely defined by this, as all other users choose the other provider. 
Further, we denote by $\mu(a,b)$ the average type of all users with type in $[a,b]$, i.e., $\mu(a,b)= \frac{\int_{a}^{b} f(\type)\type d\type}{F(b)-F(a)}$.

Each provider's costs do not only depend on how many units of his product he sells in realization, but on how many units he has to supply given a user strategy profile $\sigma$.  A provider's cost function $\costs{j}{\sigma}$ consequently is a function of the whole user strategy profile and independent of any given user's type.
 We assume that $\costs{j}{\sigma}$ is strictly increasing in the expected type of a randomly drawn user that joins provider $j$'s market under $\sigma$.
Overloading notation, we also write $\cost{i}{a}{b}$ for the cost  of provider $i$ if all users in $[a,b]$ (and no other users) choose him.

 In many applications, splitting a given population of users between two identical providers causes higher overall costs than if one provider would obtain all users, as that one provider could always provision for both sub-populations separately.
 We call such cost functions that are convex in relation to splitting the market \emph{split-convex}, i.e. $\costs{i}{}$ is  split-convex if for all $\hat{\type} \in [0,\type_{max}]$ it holds
\begin{align}
 &F(\hat{\type})\cost{1}{0}{\hat{\type}} + (1-F(\hat{\type})) \cost{1}{\hat{\type}}{\type_{\max}} \geq \cost{1}{0}{\type_{max}}.\nonumber
\end{align}
 If the inequality holds with strictly greater (i.e., $>$), we call the cost function  \emph{strictly split-convex}.

A provider's utility is his expected profit $\pi_j$ \emph{per customer in a population}, which is given by 
\begin{equation}
\pi_j(\rho_1,\rho_2, \sigma) = \int_{\type:\sigma(\type) =j} (\rho_{j}(\type) - \costs{j}{\sigma}) f(\type) d \type
\end{equation}

Given this model, we call a strategy for provider $i\in \lbrace1,2\rbrace$ an \emph{individually rational response} to any given strategy $\rho_{-i}$ of the other provider if there exists a enforceable utility-maximizing user strategy profile $\sigma$ that gives provider $i$ weakly positive profit, i.e., $\pi_i(\rho_i,\rho_{-i}, \sigma) \geq 0$.

\subsection{Equilibrium Concept}

We use the following equilibrium concept for our analysis.

\begin{definition}
        A tuple $(\rho_1,\rho_2, \sigma)$ is a Bayes-Nash equilibrium (BNE) if $\sigma$ is utility maximizing for $(\rho_1,\rho_2)$ and for $i\in \left\lbrace 1 ,2 \right\rbrace$ there exist no $\hat{\rho}_i \not = \rho_i$ and $\hat{\sigma}$ such that $\hat{\sigma}$ is utility maximizing for $(\hat{\rho}_i,\rho_{-i})$ and 
        \begin{eqnarray}
        \pi_i(\hat{\rho}_i,\rho_{-i}, \hat{\sigma}) &>& \pi_i(\rho_i,\rho_{-i}, \sigma)                         
        \end{eqnarray}
        
\end{definition} 
Note that when $\rho_1=\rho_2$, users are indifferent between all user strategy profiles, but our equilibrium definition mandates that tie breaking is done in such a way that no provider has an incentive to deviate infinitesimally only to secure himself a different user strategy profile.  
 In practice, a combination of external factors and bounded rationality imply that providers do not move their prices to a tie or even infinitesimally close to each other, at best achieving $\epsilon-$BNEs. Essentially, price differentiations that are too small are not marketable.

\section{Profit Analysis with Conservative Providers}
First, we now analyze the case where both providers are conservative, i.e., restricted to $p_1^\ell = p_2^\ell = 0$. Since they cannot split the market through pricing differences, the resulting game is similar to a classic Bertrand competition. Thus, if the providers' costs for the whole population are symmetric, they cannot extract any profit, while for non-symmetric providers, the provider with lower costs for serving the whole market can potentially extract the cost difference as a profit.

\begin{proposition}\label{PROP:CONST}
	Let both providers be conservative, i.e., $p_1^\ell = p_2^\ell = 0$. 
	W.l.o.g. assume $\cost{1}{0}{\type_{max}} \leq \cost{2}{0}{\type_{max}} $.
	Then in any BNE $(\rho_1,\rho_2,\sigma)$ the following holds: 
	\begin{eqnarray}
		p_1^f = p_2^f \in [\cost{1}{0}{\type_{max}}, \cost{2}{0}{\type_{max}} ]
	\end{eqnarray}
	and 
	\begin{align}
		\pi_1(\rho_1,\rho_2,\sigma) =& p_1^f-  \cost{1}{0}{\type_{max}} \\
		\pi_2(\rho_1,\rho_2,\sigma)=& 0
	\end{align}
	
\end{proposition}
\begin{proof}
	
 Note that any  tuple $(\rho_1,\rho_2,\sigma)$ with  
 $p_1^f = p_2^f \in [\cost{1}{0}{\type_{max}}, \cost{2}{0}{\type_{max}} ]$
 and $\sigma = [0,\type_{max}]_{\rightarrow 1}$ is a BNE as neither provider has an advantageous deviation. All users already choose provider $1$, so decreasing his price only reduces his profit, while any price increase makes him lose all users. Since all users choose him, his profit is simply $ \pi_1(\rho_1,\rho_2,\sigma) = p_1^f-  \cost{1}{0}{\type_{max}}$.
 Any lower price for provider $2$ on the other hand would lead to all users choosing him, but he would make a negative profit per user. Increasing his price would have no effect on his profit, as no users choose him. Without users, he trivially makes zero profit. 
 As a special case, when  $\cost{1}{0}{\type_{max}}= \cost{2}{0}{\type_{max}}$, then, by the same argument, $p_1^f = p_2^f = \cost{2}{0}{\type_{max}}$ with  $\sigma = [0,\type_{max}]_{\rightarrow 2}$  is an equilibrium as well, with both providers obtaining zero profit.  
 
 We now show that these are the only BNEs, by showing that any other ``potential BNE'' leads to a contradiction. 
 First, note that when $p_1^f < p_2^f$, every user strictly prefers provider $1$. But there always exists a $\hat{p}_1^f$ with $p_1^f < \hat{p}_1^f< p_2^f$ for which every user still strictly prefers provider $1$, but with a higher payment. Therefore, $p_1^f=p_2^f$ has to hold in equilibrium. Now, if $p_1^f = p_2^f < \cost{1}{0}{\type_{max}}$, provider $1$ would make a loss for every user. On the other hand, if $p_1^f = p_2^f > \cost{2}{0}{\type_{max}}$, then for any $\sigma= [0,\type_{max}]_{\rightarrow i}$, the other provider $-i$ would have an advantageous deviation in any $\hat{p}_{-i}^f$ with $\cost{2}{0}{\type_{max}} < \hat{p}_{-i}^f< p_1^f=p_2^f$. Taken together, this means that there can be no other BNE than those characterized by the proposition. 
\end{proof}

Going forward, we denote these BNEs with two conservative providers as \emph{constant BNEs}.
It should be noted that while all $p_1^f=p_2^f \in [\cost{1}{0}{\type_{max}}, \cost{2}{0}{\type_{max}} ]$ are equilibrium prices, any $p_1^f=p_2^f <\cost{2}{0}{\type_{max}}$ only does not lead to a loss for provider $2$ because he obtains no users. This makes $p_1^f = \cost{2}{0}{\type_{max}}$ the only non-pathological equilibrium. 

\section{Profit Analysis with Innovative Providers}
In this section, we  analyze the profit when either only one or both providers are innovative.  
As a first step, we bound the profit in any best response by the optimal profit a provider could attain if he could unilaterally choose the tie-breaking rule between utility-maximizing user strategy profiles.
 \begin{proposition}\label{LEM:BESTRESPONSE}
        Assume provider $2$ plays strategy $\rho_2= (p_2^f,p_2^\ell)$. Then, for any $\rho_1\not= \rho_2$ and any $\sigma$ that is utility maximizing for $\rho_1,\rho_2,$ it holds that  $\pi_1(\rho_1,\rho_2,\sigma)=0$ or $\pi_1(\rho_1,\rho_2,\sigma) < \pi_1(\rho_2,\rho_2,\sigma)$. 
 \end{proposition}      
 \begin{proof}  
        First note that for any $\rho_1\not=\rho_2,$ by the linearity of the price function, there exists a unique  $\hat{\type}$ at which payments in both markets are the same, i.e., $p_1^f+\hat{\type} p_1^\ell =p_2^f+\hat{\type} p_2^\ell$. For any utility maximizing $\sigma,$ users with type below $\hat{\type}$ choose one provider and those above  the other. 
        Now consider some $\rho_1\not= \rho_2$, $\pi_1(\rho_1,\rho_2,\sigma)>0$ and $\hat{\type}$ to be the corresponding cutoff.  
        If $p_1^f \leq p_2^f$, $p_1^\ell \geq p_2^\ell$, then any utility-maximizing user strategy profile is of the form  $\sigma=[0,\hat{\type}]_{\rightarrow 1}$ and it holds
        \begin{align}
        &\pi_1(\rho_1,\rho_2, \sigma)
        =  \int_{0}^{\hat{\type}} f({\type})(p_1^f+\type p_1^\ell - \cost{1}{0}{\hat{\type}} ) d{\type} \\
        =&  \int_{0}^{\hat{\type}} f({\type})(p_2^f + \hat{\type} p_2^\ell -\hat{\type} p_1^\ell+ \type p_1^\ell - \cost{1}{0}{\hat{\type}} ) d{\type}\\
        <& \int_{0}^{\hat{\type}} f({\type})(p_2^f + \hat{\type} p_2^\ell - \cost{1}{0}{\hat{\type}} ) d{\type}\\
        =&\pi_1(\rho_2,\rho_2, \sigma)\end{align}       
        Similarly for $p_1^f \geq p_2^f$, $p_1^\ell \leq p_2^\ell$. 
 \end{proof}
 
 It directly follows that both providers play the same strategy in any BNE, even if one provider is conservative.  
 \begin{corollary}\label{COR:EQUALSTRAT}
        In any BNE it holds  that $\rho_1=\rho_2$. 
 \end{corollary}
We now separately consider the cases where either only one or both providers are willing to employ linear pricing. 
\subsection{One Provider is Innovative}
Since all providers have to play the same strategy in any BNE, only one provider being innovative  precludes the existence of a BNE unless both providers' cost functions are very close to each other.
 \begin{theorem}\label{COR:PRICEEQUI}
 Let provider $2$  be conservative, i.e., $p_2^\ell = 0$ and provider $1$ be innovative.     There exists a BNE if and only if there exists no $\hat\type \in [0,\type_{max}]$ with 
 \begin{align} 
 & \cost{1}{0}{\type_{max}}-F(\hat t) \cost{1}{0}{\hat t}\\
 \geq & (1-F(\hat t))(\cost{1}{0}{\type_{max}}- \cost{2}{0}{\type_{max}}) 
 \end{align}
 If a BNE exists, it is equal to a constant BNE. 
 \end{theorem}
\begin{proof}
By Corollary \ref{COR:EQUALSTRAT}, in any BNE both providers will use constant pricing functions and for any two constant pricing functions that are not part of a constant BNE there exist constant deviations. By Proposition \ref{LEM:BESTRESPONSE}, any response of provider $1$ is worse than him freely choosing the user strategy profile without changing his pricing function. He can enforce any user strategy profile where he only obtains users with types lower than $\hat{t}$, so any such user strategy profile has to give him weakly lower profit than taking the whole market or he will do so for some $\hat{t} < \type_{max}$. If provider $1$ gets the lower part of the market, the users that choose provider $2$  have a higher average variance, and he therefore has a higher cost per user, than if he had the whole market. Provider $2$ therefore either has to make negative profit (and thus will deviate) or can increases his profit by slightly decreasing his price $p_2^f$ and taking the whole market.   
 \end{proof}
 
 This means that usually, providers have an incentive to become innovative if their competition is conservative, but doing so and myopically optimizing profit leads to instability. 
 Instead, an innovative provider has to be aware of their market power and utilize it in order to obtain a profit increase.

Doing so, provider $1$ can guarantee himself strictly positive payoff whenever there is any interval of users for which his costs are lower than provider $2$'s costs for all users.  
\begin{theorem}\label{PROP:POSPROF}
                Let provider $2$  be conservative, i.e., $p_2^\ell = 0$. If there exists 
                $0<\bar{\type} < \type_{max}$ with
                \begin{align}
                &\cost{1}{0}{\bar{\type}} < \cost{2}{0}{\type_{max}}\\
              \text{and}\,\,\,  &\cost{1}{0}{\bar{\type}} < \cost{1}{0}{\type_{max}} < 2\cost{1}{0}{\bar{\type}},
                \end{align}   then there exists a strategy $\rho_1$  with $p_1^\ell >0$ that  guarantees provider $1$ a non-negative payoff for any $p_2^f$ and a positive payoff for any individually rational response of provider $2$.
\end{theorem}           

\begin{proof}   
        Note that for any $p_1^f <\cost{2}{0}{\type_{max}}$, $p_1^\ell>0$, provider $2$ can only achieve a positive profit if he plays $p_2^f= p_1^f+\hat{\type} p_1^\ell$ for some $\hat{\type}>0$. To guarantee a positive payoff for provider $1$ with such a $p_1^f$ it therefore suffices that 
        \begin{align}
        &\pi_1(\rho_1,\rho_2,\hat{\type}) 
        =  \int_{0}^{\hat{\type}} f({\type})(p_1^f+\type p_1^\ell - \cost{1}{0}{\hat{\type}} ) d{\type} >0     \end{align}     
        for all $\hat{\type}$. 
        For 
        $0<\bar{\type} < \type_{max}$ with
       $\cost{1}{0}{\bar{\type}} < \cost{2}{0}{\type_{max}}$
        and
        \begin{align}
        c_1(0,\bar{\type}) < c_1(0,\type_{max}) < 2c_1(0,\bar{\type}).
        \end{align}      choose $\rho_1$ such that 
        $
        p_1^f = \cost{1}{0}{\bar{\type}}
        $ 
        and $p_1^\ell = \frac{\cost{1}{0}{\bar{\type}}}{\int_{0}^{\bar{\type}} f(\type){\type} d\type}$. 
        Then  it holds for $\hat{\type}\leq \bar{\type}$:
        \begin{align} 
       \pi_1(\rho_1,\rho_2,\hat{\type}) = &\int_{0}^{\hat{\type}} f({\type})(p_1^f+\type p_1^\ell - \cost{1}{0}{\hat{\type}} ) d{\type} \\
        >&  \int_{0}^{\hat{\type}} f({\type})(p_1^f+\type p_1^\ell - \cost{1}{0}{\bar{\type}} ) d{\type} \\   
        =&  \int_{0}^{\hat{\type}} f({\type})(\type p_1^\ell) d{\type} \\
        >&0
        \end{align}
%
   For $\hat{\type} > \bar{\type}$ it holds 
       \begin{align} 
      &\pi_1(\rho_1,\rho_2,\hat{\type})\\
       =&\int_{0}^{\hat{\type}} f({\type})(p_1^f+\type p_1^\ell - \cost{1}{0}{\hat{\type}} ) d{\type} \\
       =&  \int_{0}^{\hat{\type}} f({\type}) (\type \frac{\cost{1}{0}{\bar{\type}}}{\int_{0}^{\bar{\type}} f(\type){\type} d\type} - \cost{1}{0}{\hat{\type}}-\cost{1}{0}{\bar{\type}} ) d{\type} \\   
       >& \int_{0}^{\hat{\type}} f({\type}) (\cost{1}{0}{\bar{\type}} - \cost{1}{0}{\hat{\type}}-\cost{1}{0}{\bar{\type}} ) d{\type} \\
       =& \int_{0}^{\hat{\type}} f({\type}) (2\cost{1}{0}{\bar{\type}}-\cost{1}{0}{\hat{\type}}  d{\type} \\    
       >&0
       \end{align}
       Therefore, $\rho_1$ guarantees provider $1$ positive profit for $\sigma = [0,\hat{\type}]_{\rightarrow 1}$ with any $\hat{\type} \in [0,\type_{max}]$.
\end{proof}     

This means that provider $1$ can enforce positive profit, which is especially attractive if he could only obtain zero profit as a conservative provider. As we can see from the proof, a possible such pricing strategy with positive payoff is given by 
        $p_1^f = \cost{1}{0}{\bar{\type}}$
        and $p_1^\ell = \frac{\cost{1}{0}{\bar{\type}}}{\int_{0}^{\bar{\type}} f(\type){\type} d\type}$.
        
        This still leaves the question whether a provider that obtained positive profit as a conservative provider should also become innovative. As the following proposition shows, the answer is usually yes. 
\begin{theorem}\label{PROP:BETTERPROF}
        \todo[inline]{Note:Needs that cost is strictly increasing at $(0,sigma_{max})$, doublecheck model in the end}
                        Let provider $2$  be conservative, i.e., $p_2^\ell = 0$.
                        If provider $1$ obtains strictly positive profit in some constant BNE then there exists a strategy $\rho_1$  with $p_1^\ell >0$ that, for any individually rational response of provider $2$, guarantees provider $1$ greater profit than in any constant BNE.

\end{theorem}
\begin{proof}
         Recall that by Proposition \ref{PROP:CONST}, if provider $1$ has positive profit in some constant equilibrium then
         \begin{align}
         \cost{1}{0}{\type_{max}} < \cost{2}{0}{\type_{max}}.
         \end{align}
        Let $\bar{\type}$ be a type such that 
        \begin{align}
        &\bar{\type}
        =argmin_{\type} \frac{\cost{2}{\type}{\type_{max}}-\cost{2}{0}{\type_{max}}}{\type}
        \end{align}
        For any $\epsilon>0$ with $\epsilon< \mu(0,\type_{max})  \frac{\cost{2}{\bar{\type}}{\type_{max}}-\cost{2}{0}{\type_{max}})}{\bar{\type}}$,  let $\rho_1= (p_1^f,p_1^\ell)$
        with
        $p_1^f = \cost{2}{0}{\type_{max}}-\epsilon$
        and $p_1^\ell =\frac{\cost{2}{\bar{\type}}{\type_{max}}-\cost{2}{0}{\type_{max}}}{\bar{\type}}$.
Then for any $p_2^f$ provider $2$ obtains no users or there exists a $\type>\hat{\type}$ such that he obtains all users with type $\type>\hat{\type}$ and his profit is given by 
        \begin{align} 
        &\pi_2(\rho_1,\rho_2, \sigma)= \int_{\hat{\type}}^{\type_{max}} f({\type})(p_2^f- \cost{2}{\hat{\type}}{\type_{max}} ) d{\type} \\
        =       &\int_{\hat{\type}}^{\type_{max}} f({\type})(p_1^f+\hat{\type} p_1^\ell - \cost{2}{\hat{\type}}{\type_{max}} ) d{\type} \\
                \leq &\int_{\hat{\type}}^{\type_{max}} f({\type})(p_1^f+\hat{\type} \frac{\cost{2}{\hat{\type}}{\type_{max}}-p_1^f+\epsilon}{\hat{\type}}\\
                &- \cost{2}{\hat{\type}}{\type_{max}} ) d{\type} \\
                =& 0- (F(\type_{max})-F(\hat{\type}))\epsilon
        \end{align}
    
        Therefore, provider $2$ has no individually rational response for which he obtains any users.   
   
        It follows that for any rational response of provider $2$ and any  $\epsilon< \mu(0,\type_{max})  \frac{\cost{2}{\bar{\type}}{\type_{max}}-\cost{2}{0}{\type_{max}}}{\bar{\type}}$ provider $1$'s profit is 
               \begin{align} 
               &\pi_1(\rho_1,\rho_2,\sigma)\\
               =&\int_{0}^{\type_{max}} f({\type})(p_1^f+\type p_1^\ell - \cost{1}{0}{\type_{max}} ) d{\type} \\   
               =& \cost{2}{0}{\type_{max}}- \cost{1}{0}{\type_{max}}
               -\epsilon\\
               & + \mu(0,\type_{max})  \frac{\cost{2}{\bar{\type}}{\type_{max}})-\cost{2}{0}{\type_{max}}}{\bar{\type}}\\
               >&\cost{2}{0}{\hat{\type}}- \cost{1}{0}{\hat{\type}}
               \end{align}
               By Proposition \ref{PROP:CONST}, the profit therefore is greater than the profit in any constant BNE.
        \end{proof}

\subsection{Both Providers are Innovative}
Once one provider starts to employ linear pricing, the other provider might at some point also want to follow. 
Consequently, we now look at the case where both providers are innovative. When both providers employ linear pricing, the first provider loses much of the additional power he had when the other provider stayed conservative. Consequently, there is no general guarantee that he can still improve his profit. But as long as the cost functions are strictly split-convex, he can still guarantee himself that profits do not decrease compared to any constant BNE.   

\begin{theorem}Assume the cost function of provider $2$ is strictly split-convex.   
         Then there exists a strategy $\rho_1$  with $p_1^\ell >0$ such that, for any individually rational response of provider $2$, guarantees provider $1$ greater or equal profit than in any constant price BNE. 
        \todo[inline]{Note: This statement is weaker than saying 'in any equilibrium'. There can still be equilibria where the profit of Prov $1$ decreases!}
\end{theorem}                           

\begin{proof}   
        If provider $1$ obtains $0$ profit in all constant BNEs, nothing has to be shown.  
Otherwise, assume provider $1$ plays $\rho_1=(p_1^f,p_1^\ell)$ with $p_1^\ell =         \frac{d}{d\type}|_{\type_{max}} \cost{2}{0}{\type_{max}}$ and $
                        p_1^f+p_1^\ell \mu(0,\type_{max})  = \cost{2}{0}{\type_{max}}$.

By Proposition \ref{LEM:BESTRESPONSE}, we know that for all $\rho_2$ with utility-maximizing $\sigma=[0,\hat{\type}]_{\rightarrow 2}$ it holds that $\pi_2(\rho_1,\rho_2,\sigma)\leq \pi_2(\rho_1,\rho_1,\sigma)$ and from strict convexity it follows 
        \begin{align}
        &\cost{2}{0}{\hat{\type}}\\
         > &\cost{2}{0}{\type_{max}} \\
         &+ \frac{d}{d\type}|_{\type_{max}} \cost{2}{0}{\type_{max}} (\mu(0,\hat{\type})- \mu(0,\type_{max}))\\
         =& p_1^f+p_1^\ell \mu(0,\type_{max})+  p_1^\ell (\mu(0,\hat{\type})- \mu(0,\type_{max}))
        \end{align}
        and therefore   
                \begin{align}
                &\pi_2(\rho_1,\rho_1, \sigma)=  \int_{0}^{\hat{\type}} f({\type})(p_1^f+\type p_1^\ell - \cost{2}{0}{\hat{\type}} ) d{\type}\\
                < & \int_{0}^{\hat{\type}} f({\type})(p_1^f+\type p_1^\ell -p_1^f-p_1^\ell \mu(0,\type_{max})  \\
                &-  p_1^\ell (\mu(0,\hat{\type})- \mu(0,\type_{max})) ) d{\type}\\
                =& \int_{0}^{\hat{\type}} f({\type})(\type p_1^\ell -  p_1^\ell \mu(0,\hat{\type}) ) d{\type}\\
                =&0             
                \end{align}
Similarly for $\sigma = [0,\hat{\type}]_{\rightarrow 1}$. Therefore, any individually rational response of provider $2$ has to guarantee user strategy profile $\sigma=[0,\type_max]_{\rightarrow 1}$, for which provider $1$ has profit $\pi_1(\rho_1, \rho_2,\sigma)= \cost{1}{0}{\type_{max}}-\cost{2}{0}{\type_{max}}$, i.e. by Proposition \ref{PROP:CONST} the highest possible profit for any constant BNE. 
\end{proof}

On the other hand, if both providers are symmetric and their costs are split-convex, moving to linear prices cannot lead to any profit in equilibrium.

\begin{proposition}Assume providers are symmetric, i.e., $\costs{1}{\cdot}=\costs{2}{\cdot}$, and costs are split-convex.
	Then there can be no BNE with strictly positive profit for either provider. 
\end{proposition}                               
\begin{proof}  
If, w.l.o.g., all users prefer provider $1$ and he obtains strictly positive payoff, provider $2$ obtains zero profit but can deviate to $p_2^f=p_1^f-\epsilon$, $p_2^\ell=p_1^\ell$ to obtain strictly positive payoff for $\epsilon>0$ small enough. 
Therefore, in any potential BNE, the market is split between providers or both obtain zero profit.
Assume $(\rho_1,\rho_2,\sigma)$ is a BNE and w.l.o.g. $\sigma =[0,\bar{\type}]_{\rightarrow 1}$ for some $0<\hat{\type}<\type_{max}$. 
By Corollary 1  we can assume $\rho_1=\rho_2$.  
Then, for any $\epsilon>0$ and $\hat{\rho}_2= (p_2^f-\epsilon, p_2^\ell)$ all users prefer provider $2$, resulting in profit 
\begin{align}
&\pi_2(\rho_1,\hat{\rho}_2, [0,\type_{max}]_{\rightarrow 2})\\ 
=&  \int_{0}^{\type_{max}} f({\type})(p_2^f-\epsilon+\type p_2^\ell - \cost{2}{0}{\type_{max}} ) d{\type} \\
=&  \int_{0}^{\type_{max}} f({\type})( p_2^f+\type  p_2^\ell) d{\type} -\epsilon- \cost{2}{0}{\type_{max}}  \\
\geq &  \int_{0}^{\type_{max}} f({\type})( p_2^f+\type  p_2^\ell)  d{\type}
-\epsilon\\
&-(1-F(\hat{\type})) \cost{2}{\hat{\type}}{\type_{max}} + F(\hat{\type})\cost{2}{0}{\hat{\type}}\\
=& \pi_2(\rho_1,\rho_2,\sigma) + \pi_1(\rho_1,\rho_2,\sigma) -\epsilon
\end{align}
If $\pi_2(\rho_1,\rho_2,\sigma)>0$ or $\pi_1(\rho_1,\rho_2,\sigma)>0$, it follows that for $\epsilon$ small enough, $\pi_2(\rho_1,\hat{\rho}_2, \hat{\sigma})>\pi_2(\rho_1,\rho_2, \sigma)$, contradicting our assumption that $(\rho_1,\rho_2,\sigma)$ is a BNE. Thus, it must hold that $\pi_2(\rho_1,\rho_2,\sigma)=0$ and $\pi_1(\rho_1,\rho_2,\sigma)=0$.
	
\end{proof}

When cost functions are not split-convex, symmetric providers still always obtain the same profit in any BNE, even though it can be positive. 
\begin{proposition}Assume providers are symmetric, i.e., $\costs{1}{\cdot}=\costs{2}{\cdot}$.
        Then, in any  BNE it holds $\pi_1(\rho_1,\rho_2,\sigma)=\pi_2(\rho_1,\rho_2,\sigma)$.
\end{proposition}                               
\begin{proof}   
    Follows directly by noting that for symmetric providers, whoever has lower profits could switch users with the other provider by  decreasing prices infinitesimally. 
\end{proof}
By definition, whether a tuple $(\rho_1,\rho_2,\sigma)$ is a BNE is decided via a three-dimensional condition space, as the profit has to be better than the profit for any other tuple $(\hat{\rho_1},\hat{\rho_2},\hat{\sigma})$. This makes it very hard to evaluate whether a given tuple is a BNE. The following theorem instead characterizes equilibria by a one-dimensional condition space, greatly reducing the complexity of checking candidate equilibria. 
\begin{theorem}\label{THEO:COND}
        A tuple $(\rho_1,\rho_2,[0,\hat{\type}]_{\rightarrow i})$
     is a BNE if and only if $\rho_1=\rho_2$ and        
        \begin{align} 
        &F(\hat{\type}) \cost{i}{0}{\hat{\type}}-F(a) \cost{i}{0}{a}\\
        \leq & (F(\hat{\type})-F(a)(p_i^f+\mu(a,\type) p_i^\ell) \\
        \leq &(1-F(a))\cost{-i}{a}{\type_{max}}-(1-F(\hat{\type})) \cost{-i}{\hat{\type}}{\type_{max}}
        \end{align}
        for all $0 \leq a \leq \type_{max}$.
        If the providers are not symmetric, it also has to hold for all $0 \leq a \leq \type_{max}$
                                \begin{align} 
         &F(a) (p_i^f+\mu(0,a) p_i^\ell - \cost{-i}{0}{a})\\
         \leq&(1-F(\hat{\type})) (p_i^f+\mu(\hat{\type},\type_{max}) p_i^\ell - \cost{-i}{\hat{\type}}{\type_{max}})
                \end{align}
                        and
                                \begin{align} 
                                                        &       (1-F(a)) (p_i^f+\mu(a,\type_{max}) p_i^\ell - \cost{i}{a}{\type_{max}})\\   
                                                        \leq&F(\hat{\type}) (p_i^f+\mu(0,\hat{\type}) p_i^\ell - \cost{i}{0}{\hat{\type}})
                                                        \end{align}
\end{theorem}
\begin{proof}
        A tuple $(\rho_1,\rho_2,[0,\hat{\type}]_{\rightarrow i})$ is a BNE if no provider has an advantageous deviation. From Proposition \ref{LEM:BESTRESPONSE} we know that any deviation with a different price vector is worse than keeping the same price vector and choosing the deviating provider's most-preferred utility-maximizing user strategy profile. 
        Thus, the tuple is a BNE if and only if moving to any different user strategy profile  (weakly) decreases both providers' profits. W.l.o.g. assume $i=2$.
       Then changing to any any $[0,a]_{\rightarrow i}$ with $0\leq a < \hat{\type}$ yields a profit change for Provider $1$ of 
       \begin{align} 
       &       (1-F(\hat{\type})) (\cost{1}{\hat{\type}}{\type_{max}}- \cost{1}{a}{\type_{max}}\\
       &+ \int_{a}^{\hat{\type}} f({\type})(p_1^f+\type p_1^\ell - \cost{1}{a}{\type_{max}} ) d{\type} \\
       =&      (1-F(\hat{\type})) \cost{1}{\hat{\type}}{\type_{max}}-  (1-F(a))\cost{1}{a}{\type_{max}}\\
       &+  (F(\hat{\type})-F(a)) (p_1^f+\mu(a,\type) p_1^\ell) 
       \end{align}
       while with $\hat{\type}< a <\type_{max}$ it yields a change of
       \begin{align} 
       &       (1-F(a))(\cost{1}{\hat{\type}}{\type_{max}}- \cost{1}{a}{\type_{max}}\\
       &- \int_{\hat{\type}}^{a} f({\type})(p_1^f+\type p_1^\ell - \cost{1}{\hat{\type}}{\type_{max}} ) d{\type} \\
       =&      (1-F(\hat{\type})) \cost{1}{\hat{\type}}{\type_{max}}\\
       &-      (1-F(a))\cost{1}{a}{\type_{max}}\\
       &+  (F(\hat{\type})-F(a)) (p_1^f+\mu(a,\type) p_1^\ell) 
       \end{align}      
       Similarly, the profit change for provider $2$ is given by
       \begin{align} 
       &       F(a) \cost{2}{0}{\hat{\type}}\\
       &- \int_{a}^{\hat{\type}} f({\type})(p_1^f+\type p_1^\ell- \cost{2}{0}{\hat{\type}} ) d{\type} \\
       =& F(\hat{\type}) \cost{2}{0}{\hat{\type}}-F(a) \cost{2}{0}{a}\\
       &- (F(\hat{\type})-F(a)) (p_1^f+\mu(a,\type) p_1^\ell) 
       \end{align}
       for $a < \hat{\type}$
       and 
       \begin{align} 
       &       F(\hat{\type}) \cost{2}{0}{\hat{\type}}\\
       &+ \int_{\hat{\type}}^{a} f({\type})(p_1^f+\type p_1^\ell \- \cost{2}{0}{a} ) d{\type} \\
       =& F(\hat{\type}) \cost{2}{0}{\hat{\type}}-F(a) \cost{2}{0}{a}\\
       &-( F(\hat{\type})-F(a)) (p_1^f+\mu(a,\type) p_1^\ell) 
       \end{align}
      for $a > \hat{\type}$. Bounding all of these expressions above by zero yields the first half of the theorem.
       
               Equivalently, it has to hold for user strategy profiles that switch which provider obtains the low variance users (i.e. $[0,a]_{\rightarrow -i}$):
               \begin{align} 
               &  \int_{0}^{a} f({\type})(p_1^f+\type p_1^\ell- \cost{1}{0}{a} ) d{\type}\\
               &- \int_{\hat{\type}}^{\type_{max}} f({\type})(p_1^f+\type p_1^\ell - \cost{1}{\hat{\type}}{\type_{max}} ) d{\type} \\
               =&  F(a) (p_1^f+\mu(0,a) p_1^\ell - \cost{1}{0}{a})\\
               &-      (1-F(\hat{\type})) (p_1^f+\mu(\hat{\type},\type_{max}) p_1^\ell - \cost{1}{\hat{\type}}{\type_{max}})\\
               \leq&0
               \end{align}
  and 
                       \begin{align} 
                       & \int_{a}^{\type_{max}} f({\type})(p_1^f+\type p_1^\ell - \cost{2}{a}{\type_{max}} ) d{\type} \\
                       &- \int_{0}^{\hat{\type}} f({\type})(p_1^f+\type p_1^\ell - \cost{2}{0}{\hat{\type}} ) d{\type}\\
                       =&      (1-F(a)) (p_1^f+\mu(a,\type_{max}) p_1^\ell - \cost{2}{a}{\type_{max}})\\
                       &-F(\hat{\type}) (p_1^f+\mu(0,\hat{\type}) p_1^\ell - \cost{2}{0}{\hat{\type}})\\
                       \leq&0
                       \end{align}
                       
               If the providers are symmetric, this is equivalent to the conditions with automatically satisfied if the condition without switching is satisfied.                
                               
\end{proof}

While this fully characterizes all BNEs, it is a very technical characterization. While it can be used to \emph{check} whether a given tuple $(\rho_1,\rho_2,\sigma)$ \emph{is} a BNE, it does not enable an easy \emph{search procedure} for finding candidate BNEs. The following corollary helps with that, identifying a small subset of user strategy profiles that can be part of a BNE and reducing the search for BNEs to a one-dimensional search.
\begin{corollary}\label{COR:BNE}
If a tuple $(\rho_1,\rho_2,\sigma)$ with $0<\hat{\type}<\type_{max}$ is a BNE and the cost functions are differentiable, then it holds
        
        \begin{align}
                &\frac{d}{d \type}|_{\hat{\type}} F(\type) \cost{2}{0}{\type}\\
                                = & f(\hat{\type})(p^f+\hat{\type}p^\ell) \label{EQ:BNEPRICE}.\\
                = &\frac{d}{d \type}|_{\hat{\type}}-(1-F(\type))\cost{1}{\type}{\type_{max}}
        \end{align}
        
%
%
\end{corollary}
\begin{proof}
Note that for  $ \type = \hat{\type}$ it trivially holds        
        \begin{align} 
        &F(\hat{\type}) \cost{2}{0}{\hat{\type}}-F(\type) \cost{2}{0}{\type}\\
        =& (F(\hat{\type}))-F(\type)(p_1^f+\mu(\type,\type) p_1^\ell) \\
        = &(1-F(\type))\cost{1}{\type}{\type_{max}}-(1-F(\hat{\type})) \cost{1}{\hat{\type}}{\type_{max}}.
        \end{align}
Theorem \ref{THEO:COND} further gives us that the three expressions have the same ordering for all $\type$, therefore especially for all $\type$ in any small neighborhood around $\hat{\type}$. 
Therefore, all three expressions need to have the same derivative in $\type$ at point $\type= \hat{\type}$.  \todo[inline]{Find analysis result that actually states this. It's somewhat trivial (the lines have to be tangents to each other in the point and therefore have the same derivative) but tricky to formalize...might be possible to use Darboux's theorem}
\end{proof}
Given a cutoff point $\hat{\type}$, all potential equilibrium price functions lie on a line defined by Equation \eqref{EQ:BNEPRICE}. To find a BNE, all that remains to be done is to check whether $\hat{\type}$ with any of the  $(p^f,p^\ell)$ on that line satisfy the condition of Theorem \ref{THEO:COND}. 

\section{Welfare Analysis}

The social welfare of the market is given as the negative sum of the expected costs of both providers, i.e., $w(\rho_1,\rho_2,[0,\type]_{\rightarrow i})=- F(\hat{\type})\cost{i}{0}{\hat{\type}} - (1-F(\hat{\type})) \cost{-i}{\hat{\type}}{\type_{max}}$.
The social welfare in any constant BNE then follows directly from Proposition \ref{PROP:CONST}.
\begin{corollary}\label{LEM:CONSTWEL} Let both providers be conservative and w.l.o.g. assume $\cost{1}{0}{\type_{max}} \leq \cost{2}{0}{\type_{max}} $.
        Then the social welfare in any BNE is 
        $w(\rho_1,\rho_2,\sigma)=-\cost{1}{0}{\type_{max}}$.
\end{corollary} 
\begin{proof}   
        Follows directly from Proposition \ref{PROP:CONST}.  
\end{proof}
If only one provider is innovative, the social welfare often goes down, as the innovative provider can employ his market power to force the conservative provider to give up part of the market, even if it increases overall costs. But if both are innovative, he loses this power and in BNE, the social welfare cannot decrease compared to any constant BNE.  

\begin{proposition}W.l.o.g. assume $\cost{1}{0}{\type_{max}} \leq \cost{2}{0}{\type_{max}} $.
	The social welfare in BNE with both providers being innovative is higher than the social welfare in any constant BNE, i.e.,  
	$w(\rho_1,\rho_2,\sigma)\geq -\cost{1}{0}{\type_{max}}$.
\end{proposition}                               

\begin{proof}   
  Recall that the social welfare is given by $(-1)$ times the sum of the expected costs, i.e., $w(\rho_1, \rho_2. [0,t]_{\rightarrow i}) = (-1) \left(F(\hat{\type})\cost{i}{0}{\hat{\type}} + (1-F(\hat{\type})) \cost{-i}{\hat{\type}}{\type_{max}} \right)$.     
  Consequently, whenever all users choose provider $1$, the social welfare is the same as in any constant BNE, i.e., $w(\rho_1,\rho_2,[0,\type_{max}]_{\rightarrow 1})= -\cost{1}{0}{\type_{max}}$.
  
  We now show the claim of the proposition by contradiction.
  Assume $(\rho_1,\rho_2,[0,\hat{\type}]_{\rightarrow i})$ for $i\in {1,2}$ is an innovative BNE where the social welfare is strictly lower than the social welfare in any constant BNE. Then consequently, the social welfare under $(\rho_1,\rho_2,[0,\hat{\type}]_{\rightarrow i})$ is also strictly lower than under $(\rho_1,\rho_2,[0,\type_{max}]_{\rightarrow 1})$, i.e., the sum of the expected costs is strictly higher; formally:  
  $F(\hat{\type})\cost{i}{0}{\hat{\type}} + (1-F(\hat{\type})) \cost{-i}{\hat{\type}}{\type_{max}} > \cost{1}{0}{\type_{max}}.
  $
  
  Additionally, by Corollary 1, we can assume that $\rho_1=\rho_2$. This means that the payment of any user is independent of which provider he chooses, and therefore the sum of the revenues of both providers does not change between  $(\rho_1,\rho_2,[0,\hat{\type}]_{\rightarrow i})$ and $(\rho_1,\rho_2,[0,\hat{\type}]_{\rightarrow 1})$. Taken together, when going from the first profile to the second profile, the sum of the revenues stay the same but the costs strictly decrease, which implies that $\pi_1(\rho_1,\rho_2,[0,\type_{max}]_{\rightarrow 1}) > \pi_1(\rho_1,\rho_2,[0,\hat{\type}]_{\rightarrow i})+\pi_2(\rho_1,\rho_2,[0,\hat{\type}]_{\rightarrow i})$. Therefore,  $(\rho_1,\rho_2,[0,\hat{\type}]_{\rightarrow i})$ can not be a BNE, a contradiction.  
\end{proof}

\section{Numerical Example}
In this section, we illustrate our results via a simple numerical example. 
We assume that user types are uniformly distributed on $[0,1]$. 
Further, we assume  provider $1$ has cost function $\cost{1}{a}{b} = 0.0125+ \mu(a,b)^2$ and provider $2$ has cost function  $\cost{2}{a}{b} = 0.2+ \frac{\mu(a,b)^2}{4}$. Thus, provider $1$ has a lower cost for low types but a higher cost for high types than provider $2,$ and both providers have the same cost for the whole user population. From Proposition~\ref{PROP:CONST}, we know that when both providers are conservative, there are only zero-profit BNEs. They occur at $p_1^f= p_2^f =0.2625$ with a welfare of $-0.2625$. 
        
        \begin{figure}[t]
                \centering%
                \includegraphics[width=0.45 \textwidth]{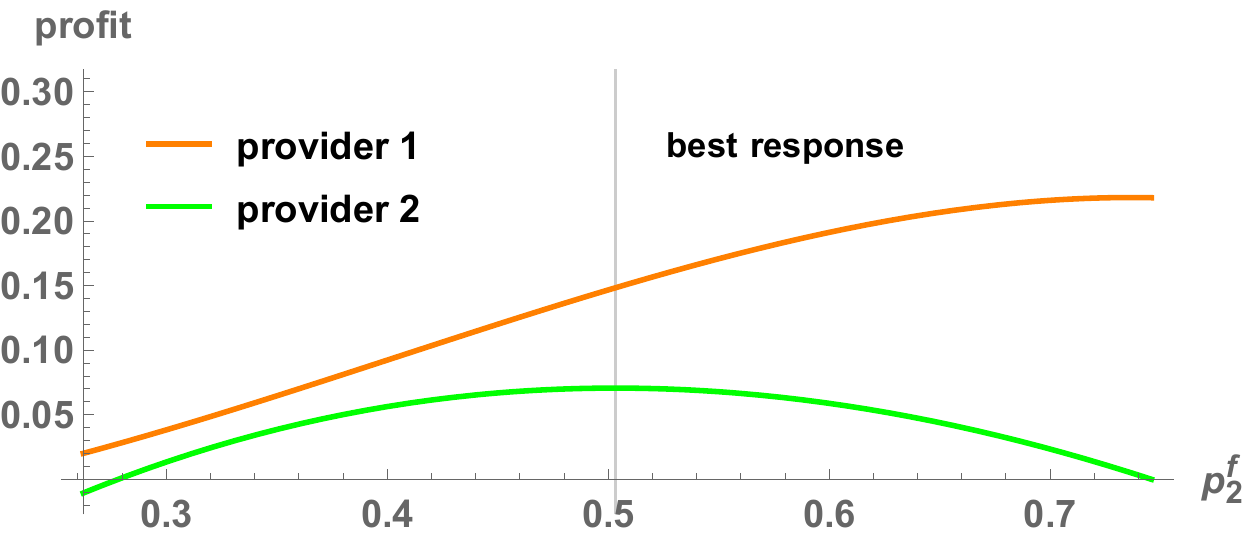}
                \caption{Profit of  both providers for all  conservative responses of provider $2$ to provider 1 playing ($p_1^f=0.215$, $p_1^\ell=0.5309$). }
                \label{FIG:1}
        \end{figure}%
        
                \begin{figure}[t]
                        \centering%
                        \includegraphics[width=0.45 \textwidth]{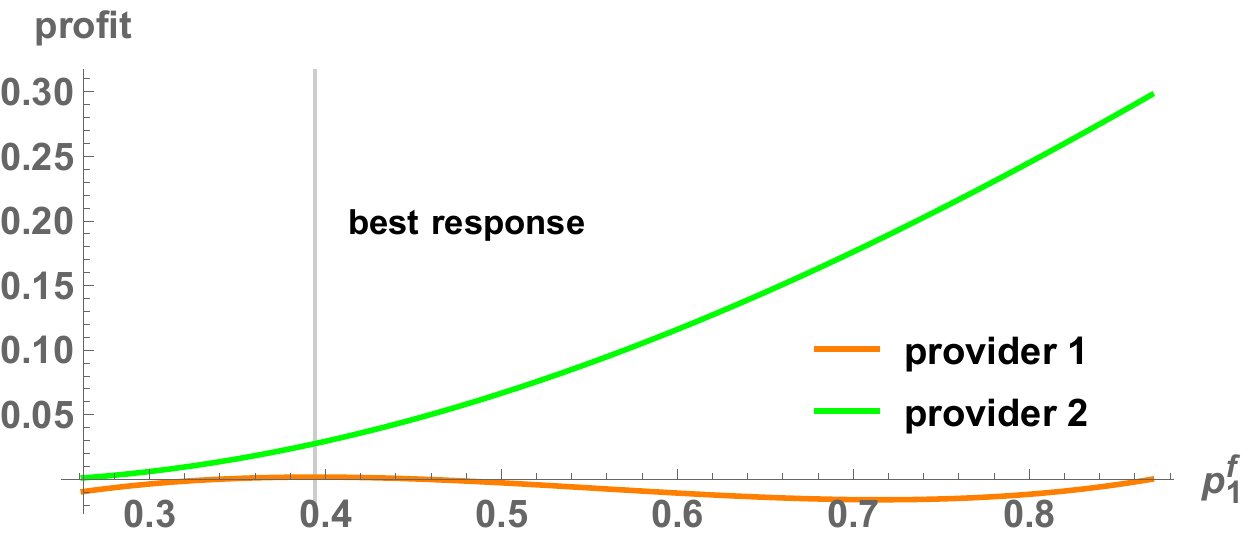}
                        \caption{Profit of both providers for all conservative responses of provider $1$ to provider 2 playing $(p_2^f=0.2506$, $p_2^\ell=0.6188)$.}
                        \label{FIG:2}
                \end{figure}%

We now consider each provider unilaterally switching to linear prices (as described in Theorem \ref{PROP:POSPROF}). Figure~\ref{FIG:1} shows the profit increase that provider $1$ could obtain by unilaterally innovating while provider $2$ remains conservative; Figure~\ref{FIG:2} shows the analogous result for when provider $2$ becomes innovative. We see that provider $1$ innovating leads to the overall better result for both providers. This is not surprising, considering that the innovative provider obtains the lower type portion of the market in which provider $1$ has lower costs. At the best response of provider $2$, the social welfare therefore increases to $-0.2077$. Nonetheless, if provider $2$ innovates instead, he can still obtain a profit of $0.0442$ at provider $1$'s best constant response of $p_1^f = 0.3941$. Since here provider $2$ uses the power of his larger strategy space as an innovative provider to force provider $1$ to obtain a high type population interval (for which provider $1$ has higher costs) social welfare unsurprisingly decreases to $-0.3846$. 

For the case where both providers are willing to employ linear pricing, Corollary \ref{COR:BNE} provides us with conditions on candidate equilibrium user strategy profiles. For our example, we can use those conditions to  find four cutoff points:  $\sigma = [0,0.595]_{\rightarrow 1}$, $\sigma= [0.5431,1]_{\rightarrow 1}$, $\sigma= [0,1]_{\rightarrow 1}$ and $\sigma = [0,0]_{\rightarrow 1}$. All of these except for  $\sigma = [0,0.595]_{\rightarrow 1}$ do not satisfy Theorem \ref{THEO:COND} and are  eliminated. For $\sigma = [0,0.595]_{\rightarrow 1}$ any $p_1^f=p_2^f \in [0, 0.0409]$, $p_1^\ell=p_2^\ell = \frac{0.2784 - p_1^f}{0.595}$ satisfy Theorem \ref{THEO:COND} and form equilibrium pricing strategies.
                \begin{figure}[t]
                        \centering%
                        \includegraphics[width=0.45 \textwidth]{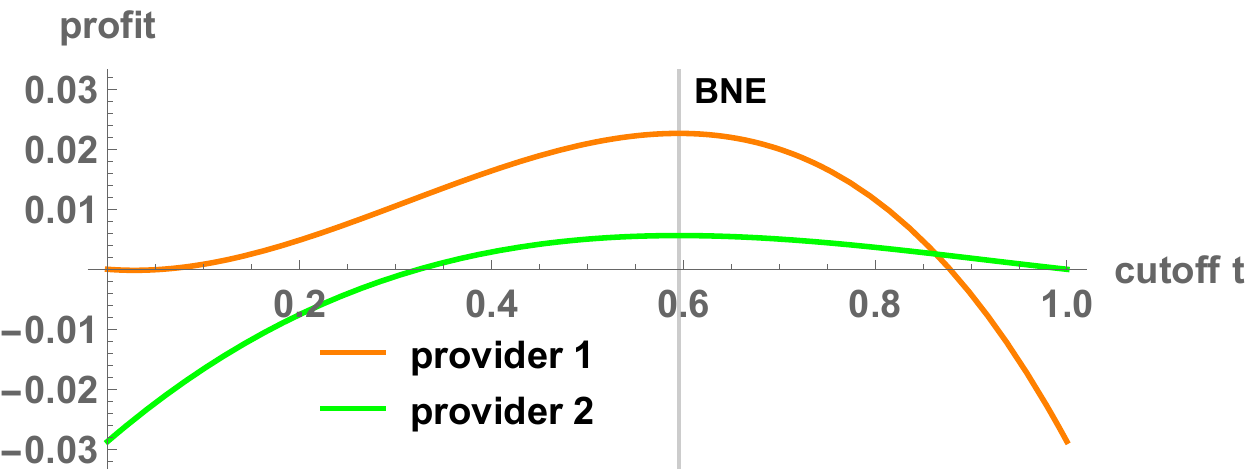}
                        \caption{Profit of both providers for different  $\sigma=[0,\type]_{\rightarrow 1}$ }
                        \label{FIG:3}
                \end{figure}%
To visualize an example BNE, Figure \ref{FIG:3} shows the profit of both providers for $p_1^f=p_2^f = 0$ and $p_1^\ell=p_2^\ell=0.4676$ given any utility maximizing $\sigma=[0,\type]_{\rightarrow 1}$. We see that neither provider wants to deviate to enforce a different user strategy profile $\sigma=[0,\type]_{\rightarrow 1}$ than  $\sigma= [0,0.595]_{\rightarrow 1}$. Social welfare at this BNE is $-0.2055$, which is even slightly better than when only provider $1$ was innovative. However, the increased competition leads to a markedly lower profit for both providers than when only provider $1$ was innovative, suggesting that it can be in a conservative provider's interest not to become innovative if the other provider is already innovative.

\section{Conclusion}
In this paper, we have studied the competitive effects of providers utilizing linear pricing rules in settings where a provider's costs depend on the average type of all his users. We have shown that, while a single provider innovating often leads to non-existence of BNEs, the innovating provider can exert the additional market power of his larger strategy space to unilaterally set prices that increase his profit for all individually rational responses of a conservative provider. We have further characterized all equilibria where both providers employ linear pricing.  We have shown that while much of this additional market power is lost once the other provider also adopts linear pricing, the increased strategy space allows providers to split the market more closely along differences in their cost functions, often increasing both providers' profits and social welfare. In conclusion, linear variance-based or type-based prices often seem superior to constant prices even under competition. Future work should study secondary effects like incentivizing users to actively lower their variance. 
\todo[inline]{Not sure if we have the space to talk much about extensions...}
 

\bibliographystyle{named}
\bibliography{ijcai20}

\end{document}